\pdfoutput=1
\documentclass[conference,compsoc]{IEEEtran}
\IEEEoverridecommandlockouts
\usepackage[nocompress]{cite}
\usepackage{amsmath,amssymb,amsthm,booktabs,graphicx,xcolor}
\usepackage{tikz}
\usetikzlibrary{arrows.meta,positioning,calc,fit}
\usepackage[hidelinks]{hyperref}
\usepackage[capitalize,noabbrev]{cleveref}

\newif\ifanonymous
\anonymousfalse

\newcommand{\Tr}{\operatorname{Tr}}
\newcommand{\E}{\mathbb{E}}
\newcommand{\cE}{\mathcal{E}}
\newcommand{\cN}{\mathcal{N}}

\newcommand{\norm}[1]{\left\lVert#1\right\rVert}
\newtheorem{theorem}{Theorem}
\newtheorem{proposition}{Proposition}
\newtheorem{corollary}{Corollary}
\newcommand{\cmark}{\ensuremath{\checkmark}}
\newcommand{\xmark}{{\color{red!70!black}\ensuremath{\times}}}
\newcommand{\pmark}{\ensuremath{\circ}}
\newcommand{\namark}{\textsc{n/a}}

\newcommand{\ResultSeeds}{200}
\newcommand{\ResultNMax}{8}

\newcommand{\StimSeeds}{20}

\newcommand{\StimIdealAgreementCases}{1200}

\newcommand{\AblationN}{8}
\newcommand{\AblationK}{1}
\newcommand{\AblationDepth}{12}

\newcommand{\AblationFullSafe}{84.5}

\newcommand{\SBCheckMaxSE}{2.9}

\newcommand{\SBCheckCases}{3}
\newcommand{\SBFactorPct}{10}

\newcommand{\SBFourteenDepth}{15}

\newcommand{\SBFourteenOneIn}{10,189}

\title{Auditing Encoder Reseeding for Quantum Error Correction under a Bounded Cloud Fault Model}
\ifanonymous
  \author{{\rm Anonymous submission}}
\else
  \author{%
    Ziqing Guo$^{1}$ \quad Anthony Lawrence$^{2}$ \quad Renyu Wang$^{2}$\\
    Randy Kuang$^{3}$ \quad Ziwen Pan$^{1}$\\[0.4em]
    $^{1}$Texas Tech University \quad $^{2}$Lightrider \quad $^{3}$Quantropi\\
    \texttt{ziqguo@ttu.edu, ziwen.pan@ttu.edu}\\
    \texttt{alawrence@lightriderinc.com, rwang@lightriderinc.com}\\
    \texttt{randy.kuang@quantropi.com}}
\fi
\date{}

\begin{document}
\maketitle

\begin{abstract}
Quantum error correction (QEC) is essential for large-scale quantum computation due to the inherent noise in physical qubits, which renders them unsuitable for algorithm execution. This issue poses a security risk for shared cloud processors, where faults may be introduced by adversaries rather than random noise. Given that the same code is used for every execution, an adversary who identifies an undetected fault can exploit it in all subsequent executions.

This study evaluates whether per-run reseeding, which involves encoding each execution with a new random circuit derived from the seed, can prevent such replay attacks. We examine two types of cloud adversaries: the provider, who compiles the circuit and thus knows the seed, and a co-tenant, who shares the processor but selects faults without seed knowledge.

To address this, we define accepted corruption as the logical alteration caused by a fault in executions that the verifier accepts. We also derive the expected accepted corruption over Haar-random encoders in closed form; this expectation is exact for uniform Clifford encoders, thereby quantifying the threat posed by any fault chosen without seed knowledge.

Our findings indicate that seed-blind reseeding reduces the proportion of executions corrupted by a replayed fault under both adversaries to a rate that decreases exponentially.
\end{abstract}


\section{Introduction}
\label{sec:intro}

Cloud platforms currently facilitate the compilation and execution of submitted circuits on remote quantum processors~\cite{javadiabhari2024quantum}. Research prototypes have been developed to multiplex several programs on shared hardware~\cite{tao2025hyperq,das2019case,saki2020analysis}. On such platforms, an adversary may introduce physical faults into a victim circuit, for instance, by manipulating control pulses~\cite{xu2024fault} or via crosstalk from a co-located circuit~\cite{saki2020analysis,deshpande2022antivirus}. Quantum error correction (QEC) is the standard defense against physical faults~\cite{shor1995scheme,steane1996error,calderbank1996good,laflamme1996perfect}. This method encodes $k$ logical qubits into $n$ physical qubits to detect or correct a specified set of faults within a code distance $D$~\cite{knill1997theory,gottesman1997stabilizer,terhal2015quantum}. Furthermore, below-threshold logical memories have been demonstrated in hardware~\cite{aharonov1997fault,knill1998resilient,acharya2023suppressing,google2025below,dennis2002topological,fowler2012surface}.

Traditional methods, as discussed in \cite{boneh2001importance,biham1997differential,barenghi2012fault}, exploit vulnerabilities within fixed cryptographic implementations. Similarly, per-run reseeding in quantum authentication ~\cite{barnum2002authentication} ensures that a fault manifesting as a logical error in one run is likely to be detected in subsequent runs. In the Clifford authentication code ~\cite{aharonov2010interactive}, the sender applies a secret, uniformly random Clifford circuit to the data along with ancillas in the state $|0\rangle$. The receiver then inverts this circuit and rejects the run if any ancilla is disturbed. Any tampering by a party unaware of the Clifford is thus rejected with high probability. However, this assurance is contingent on the attacker not observing the encoder. On a cloud processor, the provider compiles the encoder circuit and consequently learns the seed. The less capable parties, such as a co-tenant or an external fault source, commit to a fault before the seed is drawn. Consequently, it remains uncertain what level of integrity reseeding provides, against which attackers, and at what circuit cost.

This study aimed to assess the extent to which integrity reseeding is preserved against various cloud adversaries. We conceptualized an attacker capable of introducing Pauli faults with a maximum weight of $t$ into the encoded register. This attacker may either observe the seed, akin to the provider during reseeding, or select its fault independently of the seed, similar to a co-tenant. Our inquiry focuses on whether such faults can compromise the logical output of a run that the verifier deems acceptable.

To address this question, we define accepted corruption as the probability that a fault passes the syndrome test and modifies the logical state. A fault that the verifier rejects or passes but does not alter the logical state, such as a stabilizer of the code, receives a score of zero; thus, rejection is not considered a failure. This metric facilitates a direct comparison between a seed-aware attacker, who selects the most detrimental allowed Pauli after observing the encoder, and fixed-fault transfer, where a single fault is chosen based on training seeds and replayed on separate test seeds (\cref{sec:threat}). We establish the distance barrier for this metric (\cref{prop:distance}) and derive its closed-form expectation over Haar-random encoders (\cref{thm:haar}). Because the proof relies on the second moment of the code projector, the expectation is valid for any unitary 2-design, including uniform Clifford encoders (\cref{cor:clifford}). This results in a closed-form bound on the accepted corruption of every committed fault against uniform Clifford reseeding and informs the implementation for all other ensembles (\cref{sec:formal}). \Cref{fig:overview} illustrates the workflow, and \cref{tab:gap} compares this work with prior approaches.

\begin{figure*}[t]
\centering
\includegraphics[width=\textwidth]{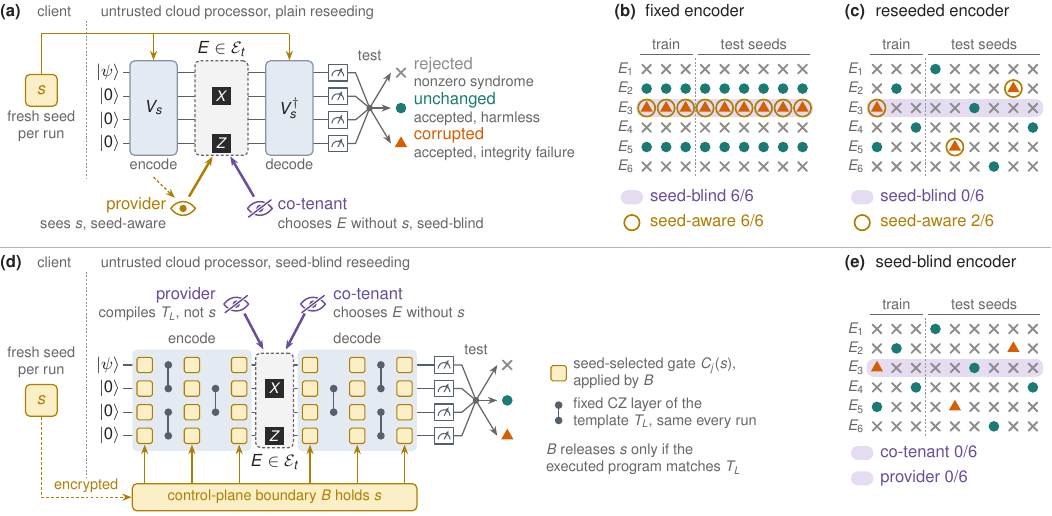}
\caption{Integrity audit and seed-blind reseeding.
(a) Plain reseeding.  A client draws a fresh seed $s$, and the provider compiles the encoder $V_s$ and therefore sees $s$.
A co-tenant chooses its fault $E$ without $s$.
The syndrome test separates rejection, harmless acceptance, and accepted corruption.
(b) On a fixed encoder the same fault is corrupted on every seed, so both attackers score 6 of 6.
(c) On a plainly reseeded encoder the fault learned on training seeds scores 0 of 6, while the seed-aware provider still scores 2 of 6.
(d) Seed-blind reseeding (\cref{eq:template}).  The provider compiles only the fixed CZ template $T_L$, and the control-plane boundary $B$ receives $s$ and applies the seed-selected single-qubit Cliffords $C_j(s)$.
(e) Neither adversary sees $s$, so both are limited to a fault learned on training seeds and score 0 of 6.}
\label{fig:overview}
\end{figure*}

\begin{table*}[t]
  \centering
  \caption{Comparison with prior approaches to randomness and fault
  robustness in QEC\@.}
  \label{tab:gap}
  \small
  \setlength{\tabcolsep}{3pt}
  \begin{tabular}{@{}lccccccc@{}}
  \toprule
  Approach
    & \shortstack{Poly.\\cost}
    & \shortstack{Per-run\\reseed}
    & \shortstack{Fault-inj.\\adversary}
    & \shortstack{Aware/blind\\split}
    & \shortstack{Reject vs.\\corrupt}
    & \shortstack{Worst\\fault}
    & \shortstack{Gate-level\\evidence} \\
  \midrule
  Fixed stabilizer QEC~\cite{knill1997theory,laflamme1996perfect}
    & \cmark & \xmark & \xmark & \xmark & \pmark & \cmark & \cmark \\
  Haar-random codes~\cite{ma2025haar,sommers2025spectral}
    & \xmark & \cmark & \xmark & \xmark & \xmark & \xmark & \xmark \\
  Unitary designs / random circuits~\cite{dankert2009designs,webb2016clifford,brandao2016designs}
    & \cmark & \cmark & \xmark & \xmark & \xmark & \xmark & \xmark \\
  Quantum authentication~\cite{barnum2002authentication}
    & \cmark & \cmark & \cmark & \pmark & \cmark & \cmark & \xmark \\
  Fault injection on quantum hardware~\cite{xu2024fault,deshpande2022antivirus,oliveira2022qufi}
    & \namark & \namark & \cmark & \xmark & \xmark & \pmark & \cmark \\
  Randomized encodings for classifiers~\cite{gong2024randomized}
    & \cmark & \cmark & \pmark & \xmark & \xmark & \xmark & \xmark \\
  \midrule
  \textbf{This work}
    & \cmark & \cmark & \cmark & \cmark & \cmark & \cmark
    & \cmark \\
  \bottomrule
  \end{tabular}

  \vspace{3pt}
  \begin{minipage}{.97\textwidth}
  \footnotesize
  The symbol \cmark{} denotes a provided property, \pmark{} a partial one
  (discussed in \cref{sec:related}), \xmark{} an absent one, and \namark{} a
  property that does not apply.  The aware/blind split requires seed-aware
  and fixed-fault scores on disjoint seeds.  Reject vs.\ corrupt requires the
  metric to distinguish detected-and-rejected faults from faults that corrupt
  accepted results.  Worst fault denotes an explicit maximum over individual
  allowed faults.
  \end{minipage}
\end{table*}

The remainder of this paper is organized as follows.  \Cref{sec:related}
reviews related work and provides the notation, and \cref{sec:threat}
defines the system and threat model.  \Cref{sec:formal} introduces the
integrity metric and its Haar reference.  \Cref{sec:encoders} describes the
encoder ensembles and the audit protocol, and \cref{sec:method} the
experimental methodology.  \Cref{sec:results} presents the results, and \cref{sec:conclusion}
concludes. The artifact is available at
\url{https://anonymous.4open.science/r/usenix_qec_hse-32FE/README.md}, and
Appendix~\ref{app:open-science} describes its contents and release
process.

\section{Related Work and Preliminaries}
\label{sec:related}

Our cloud security auditing model has four components: a seeded encoder,
an injected fault, a verifier that accepts or rejects each execution, and
the party that holds the seed.  We first review prior work on each
component and then provide the notation that the audit uses.

\subsection{Related Work}

In the context of encoders, it has been demonstrated that random encoders produce effective codes~\cite{ma2025haar,sommers2025spectral,brown2013short,gullans2021quantum}. Clifford or local random circuits can efficiently simulate Haar-random encoders~\cite{iten2016isometries,dankert2009designs,webb2016clifford,zhu2017clifford,harrow2009random,brandao2016designs,haferkamp2022random}. These findings, akin to approximate recovery~\cite{beny2010general,barnum2002reversing}, pertain to average performance metrics. Regarding an adversary with knowledge of the encoder, studies on side-channel and fault-injection have revealed vulnerabilities in quantum hardware~\cite{erata2024quantum,oliveira2022qufi}, while existing defenses safeguard hardware identity, control pulses, classifiers, or performance under benign noise~\cite{phalak2021quantum,trochatos2023quantum,gong2024randomized,wallman2016noise,hashim2021randomized,sivak2026reinforcement}. For verifiers, advancements in quantum authentication~\cite{broadbent2016simulation,portmann2017recycling} and verifiable delegated computation~\cite{broadbent2009universal,fitzsimons2017unconditionally,mahadev2018classical,gheorghiu2019verification} offer protection for computations on untrusted servers. Concerning seed holders, seed standards are designed to protect the seed from external threats~\cite{nist2019fips140,barker2015drbg,qpp-rng-sci-kuang-2025,kuang-rpss-pre-2026,nist2024fips203,nist2024fips204,nist2024fips205,nist2024ir8547}, yet attacker-selected randomness can undermine randomized defenses~\cite{dahiya2024randomness}. In a cloud processing environment, the provider compiles the encoder and thus has access to the seed. Particularly, blind signatures enable a signer to authenticate a message without viewing it~\cite{bouillaguet2025blinding}, and seed-blind reseeding extends this concept to Quantum Error Correction (QEC), which presents additional challenges as the seed determines the gates that the provider must compile.

\subsection{Preliminaries}
We provide the components for cloud quantum model shown in \cref{fig:overview}.
\paragraph{Encoder} The client encodes $k$ logical qubits into $n$ physical qubits, denoting the logical and physical Hilbert-space dimensions as $K=2^k$ and $N=2^n$, respectively. A seed $s$ determines an $n$-qubit unitary encoder $U_s$. When applied to the logical state along with $n-k$ ancillas in $|0\rangle$, it defines an isometry $V_s$ from $\mathbb{C}^K$ to $\mathbb{C}^N$, satisfying $V_s^\dagger V_s=I_K$, where $I_K$ is the $K\times K$ identity matrix, and a rank-$K$ code-space projector $P_s=V_sV_s^\dagger$. For a complex matrix $M$, $M^\dagger$ represents its conjugate transpose, $\Tr(M)$ its trace, and $\norm{M}_F^2=\Tr(M^\dagger M)$ its squared Frobenius norm.

\paragraph{Fault} The adversary introduces a Pauli string $E=E_1\otimes\cdots\otimes E_n$ with $E_j\in\{I,X,Y,Z\}$ between encoding and verification. Its support $\operatorname{supp}(E)=\{j\mid E_j\neq I\}$ is the set of physical qubits that it affects, its weight $\operatorname{wt}(E)=|\operatorname{supp}(E)|$ counts them, and \cref{sec:threat} constrains this weight by a budget $t$.

\paragraph{Verifier} The verifier accepts an execution if the faulty state resides in the code space and rejects it otherwise (\cref{sec:threat} provides the circuit). For a stabilizer code, the joint $+1$ eigenspace of an abelian group $\mathcal{S}$ of Pauli strings that excludes $-I$~\cite{gottesman1997stabilizer}, a fault is categorized into one of three classes. A fault that anticommutes with an element of $\mathcal{S}$ maps the code space onto an orthogonal space and is rejected. A fault within $\mathcal{S}$ up to a phase is accepted and benign. A fault in the normalizer $N(\mathcal{S})$ but outside $\mathcal{S}$ is accepted and modifies the logical state, with the code distance $D$ being the minimum weight of such a fault. A code can correct a list of faults $\{E_a\}$ if and only if the Knill--Laflamme condition~\cite{knill1997theory} holds:
\begin{equation}
PE_a^\dagger E_bP=\alpha_{ab}P
\label{eq:kl}
\end{equation}
for every pair $(a,b)$ with a scalar $\alpha_{ab}$, indicating that every product of two listed faults acts on the code space as a scalar multiple of the identity, thus revealing no logical information.

\paragraph{Encoder distribution} Reseeding selects $U_s$ from a distribution. The Haar measure is the sole distribution over unitaries invariant under multiplication by any fixed unitary, serving as our analytical reference. The Clifford group comprises unitaries that map Pauli strings to Pauli strings under conjugation. It is generated by the Hadamard, phase, and controlled-NOT (CNOT) gates, and it also contains the controlled-$Z$ (CZ) gate, which applies $Z$ to one qubit when the other is $|1\rangle$, and the SWAP gate, which exchanges two qubits.  Clifford circuits allow efficient classical simulation~\cite{gottesman1997stabilizer,aaronson2004stabilizer}, and we use the stabilizer simulator Stim~\cite{gidney2021stim}. A distribution is a unitary 2-design if every expectation of a polynomial of degree two in the entries of $U$ and two in those of $U^\dagger$ matches its Haar value~\cite{dankert2009designs}. The uniform distribution over the Clifford group constitutes a 2-design~\cite{dankert2009designs,webb2016clifford}, which \cref{sec:formal} employs to transfer a Haar expectation to Clifford encoders.

\section{System and Threat Model}
\label{sec:threat}
Here, for every execution, a client selects an encoder family $\mathsf{F}$, the
widths $n$ and $k$, and a fresh integer seed $s$.  The families, detailed in
\cref{sec:encoders}, are Haar-random isometries, uniform Clifford circuits,
the fixed $[[5,1,3]]$ code, distance-preserving reseeding, which appends
seed-selected single-qubit Cliffords and a qubit permutation to that code
without lowering its distance, a Hadamard-structured ensemble (HSE) built from
Hadamard, phase, permutation, and CNOT layers, whose depth $d$ counts its
mixing layers, and the seed-blind template introduced below.  The
corresponding encoder $U_s$ defines the isometry $V_s$ of
\cref{sec:related}.  A physical fault acts after encoding and before ideal
projection, recovery, or inverse decoding.  We restrict faults to this
encoded window for every adversary, because a fault on the bare logical
state before encoding or after decoding changes the output under every
encoder, and protecting against it requires verifiable delegated
computation~\cite{broadbent2009universal,fitzsimons2017unconditionally,mahadev2018classical}.
We generate the Haar isometries
and compute every metric of \cref{sec:formal} by explicit matrix computation
(the matrix audit of \cref{sec:encoders}).

There are state preparation, encoding, injection,
inverse decoding, and measurement running on the remote quantum processor.
Here inverse decoding is applying $U_s^\dagger$ with the same seed used
for encoding.  Measuring the $n-k$ decoded ancillas implements the syndrome
test, so that all-zero ancillas accept and any nonzero ancilla rejects.  In
the equivalent matrix calculation, projection onto the code space supplies
the same accept-or-reject split.  Here, we refer an attacker seed-aware if it
observes $s$ before choosing its fault and seed-blind otherwise.  Under plain
reseeding, in which the provider compiles $U_s$, the provider learns $s$ and
is therefore seed-aware, so only an external fault source or a co-tenant
that commits before observing the victim circuit is seed-blind.

\paragraph{Seed-blind reseeding}
To remove the provider's knowledge of the seed, we split the encoder into a
public template and seed-selected single-qubit gates,
\begin{equation}
 U_s=C_L(s)\,\Lambda_L\,C_{L-1}(s)\cdots \Lambda_1\,C_0(s),
 \label{eq:template}
\end{equation}
where each $\Lambda_j$ is a fixed layer of CZ gates and each $C_j(s)$ applies an
independent uniform single-qubit Clifford to every qubit.  The provider
compiles the template $T_L=(\Lambda_1,\ldots,\Lambda_L)$ with one single-qubit slot on
every qubit between consecutive layers, so the compiled program is identical
for every seed.  The control-plane boundary $B$, which serves as the classical controller for the qubits, retains $s$ and populates each slot with the corresponding gate of $C_j(s)$ during the encoding process, and with its inverse during decoding.  Following
the cryptographic module boundary of FIPS~140-3~\cite{nist2019fips140}, we note that $B$
keeps $s$ confidential from the provider and releases it when the
executed program matches the approved template, so that the provider cannot
insert seed-dependent gates. But delivering $s$ from the client to $B$ is a
key-establishment problem that lies outside this paper.  Under these
assumptions the provider observes $T_L$ but not $s$, and
\cref{prop:view} places it in the seed-blind class.

\paragraph{Fault budget}
For an integer fault budget $t\in\{1,\ldots,n\}$, we define the enumerated
fault set as
\begin{equation}
  \cE_t=\{E\in\{I,X,Y,Z\}^{\otimes n}\mid
  1\leq\operatorname{wt}(E)\leq t\}.
  \label{eq:error-set}
\end{equation}
The relation between the budget $t$
and the distance $D$ determines whether reseeding can affect accepted corruption, as we
show in \cref{prop:distance}.  This model
places an ideal Pauli on a chosen qubit.  Such placement is stronger than the
control a co-tenant obtains through crosstalk, and we therefore regard it as
an upper bound on co-tenant control rather than as a device measurement.
The enumeration excludes coherent rotations, leakage out of the qubit
subspace, correlated multi-qubit faults, and measurement manipulation.  In
addition, our gate-level study adds benign depolarizing noise
(\cref{sec:method}).

\paragraph{Attacker knowledge}
The seed-aware attacker observes $s$ before choosing
$E\in\cE_t$.  This models a provider that compiles the circuit and then selects
a fault inside \cref{eq:error-set}.  Since the reconstruction of victim
circuits from controller side channels has been
demonstrated~\cite{erata2024quantum}, we also grant this attacker the
compiled circuit.  Under seed-blind reseeding, the provider keeps the
compiled template and its choice of qubit but loses $s$, so it joins the
seed-blind class while knowing $T_L$.  The seed-blind class does not observe $s$, and within
that class we score two procedures.  Fixed-fault transfer selects one fault
on training seeds and applies that same fault to every test seed.  The
placement-adaptive attacker likewise does not observe the seed.  On each
test seed it learns a fresh set of $t$ qubits, drawn independently of the
encoder, and it may use any allowed Pauli inside that set, the training
scores determining which Pauli it uses.

\paragraph{Security event}
Hereafter we denote that the security event is the accepted corruption, where an execution that passes the
verifier although the fault has changed its logical state.
\Cref{sec:formal} defines its probability $c_V(E)$ for an encoder $V$ and a
fault $E$.  For Clifford encoders and Pauli faults, it is the probability
that a shot is accepted with a wrong logical result, and for a general
encoder it is this probability averaged over all logical inputs.  Rejected
shots produce no accepted output and therefore no corruption.  However, as
they diminish the proportion of usable outputs, we report acceptance
separately. Under plain reseeding,
we score provider-controlled faults against the seed-aware attacker, and
under seed-blind reseeding against the seed-blind attacker, as defined in
\cref{prop:view}.

\section{Integrity Metric and Haar Reference}
\label{sec:formal}

The Knill--Laflamme condition (\cref{eq:kl}) counts every fault that the
code cannot correct as a failure.  In an encode-and-verify circuit,
however, a fault has three possible outcomes: it is rejected, it is
accepted without changing the logical state, or it is accepted and changes
the logical state.  Only the third outcome is a security event, because a
rejected run produces no output.  This section defines the probability of
that outcome, relates it to the Knill--Laflamme condition
(\cref{prop:decomp}), defines the score of each attacker, and computes the
expected score over random encoders (\cref{thm:haar}).

For each encoder and fault, our audit returns two numbers in $[0,1]$, namely
the acceptance $p_{\mathrm{acc}}$ and the accepted logical disturbance $c$.
Their three components are the rejection $1-p_{\mathrm{acc}}$, the harmless
accepted weight $p_{\mathrm{acc}}-c$, and the accepted-corruption weight
$c$.  Thus $(p_{\mathrm{acc}},c)=(0,0)$ means that the fault is always
detected and rejected, $(1,0)$ means that it is accepted but harmless, and
$(1,1)$ means that it is accepted as a logical change.  We define both
numbers in \cref{eq:pacc,eq:corr} below.

Let $V$ denote one encoder isometry with code projector $P$, the seed
subscript being omitted locally.  To test whether $V$ corrects the whole
fault set, we measure how far each pair of faults departs from
\cref{eq:kl}.  We define the supported operator list
$\{E_a\}_{a=0}^{q}=\{E_0=I_N\}\cup\cE_t$, where $q=|\cE_t|$ and $I_N$ is the
physical identity, and we apply \cref{eq:kl} to this list for
$0\leq a,b\leq q$.  The matrix $B_{ab}=V^\dagger E_a^\dagger E_bV$ is the
product $E_a^\dagger E_b$ as seen inside the code space, $\alpha_{ab}I_K$
with $\alpha_{ab}=\Tr(B_{ab})/K$ is its scalar part, and the residual
$r_{ab}$ measures the remainder,
\begin{equation}
\begin{aligned}
 r_{ab}&=K^{-1/2}\norm{B_{ab}-\alpha_{ab}I_K}_F,\\
 R_{\max}&=\max_{0\leq a,b\leq q}r_{ab},\\
 R_{\mathrm{rms}}&=\left[\frac{1}{(q+1)^2}
       \sum_{a,b=0}^{q}r_{ab}^2\right]^{1/2}.
\end{aligned}
 \label{eq:kl-residual}
\end{equation}
Consequently $R_{\max}=0$ is equivalent to exact correction of the
enumerated operator list, while $R_{\mathrm{rms}}$ summarizes its average
pairwise violation.  \Cref{sec:encoders} uses $R_{\max}$ to label encoders
as exactly correcting.

For a unitary physical fault $E\in\cE_t$, let $\rho_{\mathrm{code}}=P/K$ be
the maximally mixed encoded logical state, that is, the uniform mixture of
all encoded logical states.  We describe the fault as the
logical qubits see it after verification.  Its compressed logical action
$A_E$ is the fault restricted to the code space, its scalar coefficient
$\mu_E$ gives the part of $A_E$ that multiplies every logical state by the
same number and therefore changes nothing, and its traceless component
$A_E^0$ is the remainder, which changes the logical state:
\begin{align}
 A_E&=V^\dagger EV,&
 \mu_E&=K^{-1}\Tr(A_E),&
 A_E^0&=A_E-\mu_EI_K.
 \label{eq:compressed-action}
\end{align}
Projecting the faulty state back onto the code space gives the acceptance
probability, and removing the scalar component gives the accepted logical
disturbance, so that
\begin{align}
 p_{\mathrm{acc},V}(E)
   &\equiv\Tr\!\left[P E\rho_{\mathrm{code}}E^\dagger\right]
    =K^{-1}\norm{A_E}_F^2,\label{eq:pacc}\\
 c_V(E)&\equiv K^{-1}\norm{A_E^0}_F^2.
 \label{eq:corr}
\end{align}
Thus $p_{\mathrm{acc},V}(E)$ is the probability that the fault passes the
verifier, and the accepted corruption $c_V(E)$ is the probability that it
passes and changes the logical state. In the ideal Clifford and Pauli
experiment, $c_V(E)$ equals the probability that a shot is accepted with
a wrong logical result.  For a nonunitary
operator the same formulas define nonnegative weights, but our primary fault
set contains only unitary Paulis.

\begin{proposition}[Operational decomposition]
\label{prop:decomp}
For every isometry $V$ and physical operator $E$, let $r_{I,E}$ denote
\cref{eq:kl-residual} evaluated for the supported pair $(I_N,E)$.  Then
\begin{equation}
 c_V(E)=p_{\mathrm{acc},V}(E)-\frac{|\Tr(A_E)|^2}{K^2}
       =r_{I,E}^2.
 \label{eq:decomposition}
\end{equation}
Consequently $0\leq c_V(E)\leq p_{\mathrm{acc},V}(E)$, with $c_V(E)=0$ if
and only if the accepted action is scalar on the code space.
\end{proposition}
\begin{proof}
By \cref{eq:compressed-action}, $A_E=\mu_EI_K+A_E^0$ and $\Tr(A_E^0)=0$.
Hence the cross term $2\operatorname{Re}[\overline{\mu_E}\Tr(A_E^0)]$ in the
expansion of $\norm{A_E}_F^2$ vanishes, and we have
\begin{equation}
\begin{aligned}
 \norm{A_E}_F^2
 &=\norm{\mu_EI_K}_F^2+\norm{A_E^0}_F^2\\
 &=K|\mu_E|^2+\norm{A_E^0}_F^2.
\end{aligned}
 \label{eq:pythagorean}
\end{equation}
Dividing \cref{eq:pythagorean} by $K$ and substituting $\mu_E=\Tr(A_E)/K$
yields the first equality in \cref{eq:decomposition}.  For the supported
pair $(I_N,E)$, $B_{I,E}=A_E$ and $\alpha_{I,E}=\mu_E$, and substituting
these identities into \cref{eq:kl-residual} yields $r_{I,E}^2=c_V(E)$.
Nonnegativity follows from the squared norm.  Because orthogonal projection
cannot increase a norm, $c_V(E)\leq p_{\mathrm{acc},V}(E)$, and equality to
zero holds if and only if $A_E^0=0$, that is, when the accepted action is
scalar.  Applying $A_E$ to half of a normalized maximally entangled state
gives the same decomposition, in which the term $|\Tr(A_E)|^2/K^2$ is the
unchanged-state weight and $c_V(E)$ is the orthogonal weight.
\end{proof}

We emphasize that \cref{eq:decomposition} allows us to score accepted
runs and compare attackers who know the seed with attackers who do not.
We now define the score of each attacker.  The seed-aware score is
\begin{equation}
 C_{\mathrm{aware}}(V_s)=\max_{E\in\cE_t}c_{V_s}(E).
 \label{eq:aware}
\end{equation}
Here $V_s$ is the encoder generated by seed $s$, so that
$C_{\mathrm{aware}}(V_s)$ is the accepted corruption of the most harmful
allowed fault once that encoder is known.  Let $S_{\mathrm{tr}}$ and $S_{\mathrm{te}}$ be nonempty, disjoint
sets of training and test seeds.  Fixed-fault transfer selects one fault
using only $S_{\mathrm{tr}}$, freezes it, and evaluates it on
$S_{\mathrm{te}}$, namely
\begin{align}
 \widehat E&=\arg\max_{E\in\cE_t}|S_{\mathrm{tr}}|^{-1}
                  \sum_{s\in S_{\mathrm{tr}}}c_{V_s}(E),\label{eq:blind-select}\\
 C_{\mathrm{blind}}&=|S_{\mathrm{te}}|^{-1}
                  \sum_{s\in S_{\mathrm{te}}}c_{V_s}(\widehat E).
 \label{eq:blind-test}
\end{align}
Thus $\widehat E$ is the fault with the highest average corruption on the
training seeds, and $C_{\mathrm{blind}}$ is its average corruption on unseen
seeds.  For comparison on the same seeds, the held-out seed-aware score is
\begin{align}
 \overline C_{\mathrm{aware}}
   &=|S_{\mathrm{te}}|^{-1}\sum_{s\in S_{\mathrm{te}}}
       C_{\mathrm{aware}}(V_s).\label{eq:aware-test}
\end{align}

The following proposition shows that the code distance removes accepted
corruption for every fault of weight below $D$, which identifies the regime
in which reseeding can change the outcome.

\begin{proposition}[Distance bound]
\label{prop:distance}
Let $V$ encode a stabilizer code of distance $D$ and let $E$ be a Pauli
operator with $\operatorname{wt}(E)<D$.  Then $c_V(E)=0$.  Consequently,
$C_{\mathrm{aware}}(V)=0$ in \cref{eq:aware} whenever $t<D$, and a
seed-aware attacker can obtain accepted corruption only when $t\geq D$.
\end{proposition}
\begin{proof}
Either $E$ anticommutes with some element of the stabilizer group
$\mathcal{S}$ or it commutes with all of them.  In the first case $E$ maps
the code space onto an orthogonal syndrome space, so $PEP=0$, hence
$A_E=V^\dagger EV=0$ and $p_{\mathrm{acc},V}(E)=c_V(E)=0$.  In the second
case $E\in N(\mathcal{S})$, and since $\operatorname{wt}(E)<D$ the
definition of the distance forces $E$ to be a stabilizer up to a phase, so
$E$ acts on the code space as that phase times the identity.  Then
$A_E=\mu_EI_K$ with $|\mu_E|=1$, hence $A_E^0=0$ and $c_V(E)=0$ by
\cref{eq:corr}.  The statement about $C_{\mathrm{aware}}$ follows because
every $E\in\cE_t$ has weight at most $t<D$.
\end{proof}

\begin{theorem}[Haar acceptance and disturbance]
\label{thm:haar}
Let $V$ be a Haar-distributed $N\times K$ isometry and let $E$ be a fixed
traceless unitary on $\mathbb{C}^{N}$, such as any nonidentity Pauli.  Then
\begin{align}
 \E_V[p_{\mathrm{acc},V}(E)]&=\frac{KN-1}{N^2-1},\label{eq:haar-acc}\\
 \E_V[c_V(E)]&=\frac{N(K^2-1)}{K(N^2-1)}.\label{eq:haar-corr}
\end{align}
For every fixed $K\geq2$, both quantities decay as $\Theta(N^{-1})$, and
for $K=1$, $c_V(E)=0$ identically.
\end{theorem}
The proof, which we give in Appendix~\ref{app:proof}, uses only the
second moment $\E[P\otimes P]$ of a uniformly random rank-$K$ projector.
This observation yields a guarantee for the main construction of this
paper.

\begin{corollary}[Seed-blind bound for uniform Clifford reseeding]
\label{cor:clifford}
Let $V_s=U_sW$, where $U_s$ is uniform over the $n$-qubit Clifford group
and $W$ appends $n-k$ ancillas in $|0\rangle$ (\cref{sec:encoders}).  Then
\cref{eq:haar-acc,eq:haar-corr} hold with $\E_V$ replaced by the
expectation over the seed $s$.  In particular, for every fault
$\widehat E\in\cE_t$ committed before the seed is drawn,
\begin{equation}
 \E_s\bigl[c_{V_s}(\widehat E)\bigr]=\frac{N(K^2-1)}{K(N^2-1)}.
 \label{eq:clifford-bound}
\end{equation}
\end{corollary}
\begin{proof}
Write $P_s=V_sV_s^\dagger=U_sP_0U_s^\dagger$, where $P_0=WW^\dagger$ is a
fixed rank-$K$ projector.  The proof of \cref{thm:haar} evaluates
$p_{\mathrm{acc},V}(E)$ and $|\Tr(A_E)|^2$ as linear functionals of
$\E[P\otimes P]$.
The Clifford group is a unitary
2-design~\cite{dankert2009designs,webb2016clifford}, so
$\E_s[P_s\otimes P_s]$ coincides with the Haar value in
\cref{eq:projector-ansatz}, and both expectations follow.  Every
nontrivial Pauli is a traceless unitary, and a fault committed before the
draw is independent of $s$, so \cref{eq:clifford-bound} holds in
particular for the fault $\widehat E$ that \cref{eq:blind-select} selects
on training seeds disjoint from the test seeds.
\end{proof}

\begin{proposition}[Seed-independent view]
\label{prop:view}
Let an attacker observe a random variable $\Omega$ that is independent of
the seed $s$, and let it choose $E\in\cE$ as any randomized function of
$\Omega$.  Then
\begin{equation}
 \E\bigl[c_{V_s}(E)\bigr]\leq C^{\star}_{\mathrm{blind}}(\cE)
   \equiv\max_{E\in\cE}\E_s\bigl[c_{V_s}(E)\bigr].
 \label{eq:view}
\end{equation}
\end{proposition}
\begin{proof}
Conditioned on $\Omega$ and on the attacker's internal randomness, $E$ is fixed
and $s$ keeps its distribution, so the conditional expectation of
$c_{V_s}(E)$ equals $\E_s[c_{V_s}(E)]\leq C^{\star}_{\mathrm{blind}}(\cE)$.
Averaging over $\Omega$ preserves the bound.
\end{proof}
Here $C^{\star}_{\mathrm{blind}}(\cE)$ is the largest average corruption that
any single fault in $\cE$ attains over random seeds.  Under seed-blind
reseeding, $\Omega$ consists of $T_L$ and of the outcomes of
earlier executions, which rely on earlier, independent seeds, so
\cref{eq:view} bounds the provider as it bounds a co-tenant.

Uniform
Clifford encoders attain the rate of \cref{eq:haar-corr} (\cref{cor:clifford}), but
their circuit depends on $s$ and therefore reveals it to the compiler.  The
template of \cref{eq:template} keeps the circuit fixed, and the following
proposition computes its seed-blind rate.  It follows the support of the
fault as decoding pushes the fault back through the layers.  Because each
layer of random single-qubit Cliffords makes the Pauli on every supported
qubit uniformly random, the corruption depends on the fault through its
support, which evolves as a Markov chain with one transition matrix $\Gamma_j$ per CZ layer.

\begin{proposition}[Support chain]
\label{prop:chain}
Let $k=1$, with qubit $0$ logical, and let $V_s$ follow
\cref{eq:template}.  Then $\E_s[c_{V_s}(E)]$ depends only on the support
$S=\operatorname{supp}(E)$ and equals
\begin{equation}
 \rho_L(S)=w^{\top}\Gamma_1\Gamma_2\cdots \Gamma_L\,\delta_S,
 \label{eq:chain}
\end{equation}
where $\delta_S$ is the indicator vector of $S$, $w(S')=[0\in S']\,3^{1-|S'|}$
is the probability that a fault with support $S'$ is accepted as a logical
error after the last Clifford layer, and $\Gamma_j$ acts independently on
every CZ pair $(a,b)$ of $\Lambda_j$.  A support that
meets the pair only at $a$ spreads to $\{a,b\}$ with probability $2/3$ and
stays at $\{a\}$ otherwise, symmetrically for $b$, and a support that contains both qubits moves
to $\{a\}$ or to $\{b\}$ with probability $2/9$ each and stays otherwise.
The distribution $\pi(S)\propto3^{|S|}$ over nonempty supports is
stationary for every $\Gamma_j$, that is, unchanged by it, and $w^{\top}\pi=3\cdot2^{n-1}/(4^n-1)$, which is
\cref{eq:haar-corr} at $K=2$.
\end{proposition}
\begin{proof}
A uniform single-qubit Clifford maps each of $X$, $Y$, and $Z$ to each of
them with probability $1/3$.  After conjugation by $C_j(s)$, the decoded
fault is therefore uniform over the nonidentity Paulis on its support,
independently of earlier layers.  Conjugation by a CZ gate on $(a,b)$ maps
$X_a$ and $Y_a$ to operators supported on both qubits and leaves $Z_a$ in
place, which gives the first transition.  Of the nine products
$P_aQ_b$ with $P,Q\in\{X,Y,Z\}$, the images of $X_aZ_b$ and $Y_aZ_b$ are
supported on $a$ alone, those of $Z_aX_b$ and $Z_aY_b$ on $b$ alone, and the
remaining five on both, which gives the second.  Decoding conjugates by
$C_L(s)$ first and by $C_0(s)$ last, which produces $\Gamma_1\cdots\Gamma_L$ in the
stated order.  After $C_0(s)$, the fault is accepted and logical if and
only if qubit $0$ lies in the support and every supported ancilla carries $Z$.
Of the $3^{|S'|}$ equally likely Paulis on a support $S'$ that contains
qubit $0$, three have this form, which gives $w(S')$.  For
stationarity, $3^{|S|}$ nonidentity Paulis have support $S$, so
$\pi$ is the support distribution of a uniformly random nonidentity Pauli.
Conjugation by any Clifford permutes the nonidentity Paulis and therefore
preserves this uniform distribution, so $\pi$ is stationary for every
$\Gamma_j$.  Under $\pi$, the decoded fault is a uniformly random nonidentity
Pauli, and it is accepted and logical if and only if it acts on qubit $0$ by
$X$, $Y$, or $Z$ and on every ancilla by $I$ or $Z$.  Of the
$4^n-1$ nonidentity Paulis, $3\cdot2^{n-1}$ have this form, which
gives $w^{\top}\pi$.
\end{proof}
The fixed point of the chain is therefore the uniform Clifford rate of
\cref{cor:clifford}.  We evaluate \cref{eq:chain} and report the depth
$L$ at which the worst support of weight at most $t$ approaches it
(\cref{sec:method}).

\paragraph{Fixed recovery}
To verify that a reseeded encoder still corrects faults, we also apply the
transpose (Petz) recovery, a recovery map constructed from the code and the
fault set without optimization whose fidelity is close to the
optimum~\cite{barnum2002reversing,beny2010general}.  We score it by the
entanglement fidelity $F_e\in[0,1]$, the overlap between half of a
maximally entangled logical state and its image after encoding, fault, and
recovery, so that $F_e=1$ denotes perfect recovery.
Appendix~\ref{app:petz} gives the construction.

\section{Encoder Ensembles and Audit Protocol}
\label{sec:encoders}

In this section, we describe the encoder families under comparison and
the protocol by which we audit each of them. We evaluate all families using a one-two-qubit counting rule, where a two-qubit instruction is defined as either one CNOT or one SWAP in the source circuit. This evaluation occurs prior to the routing process, which is the compilation step that incorporates SWAP gates to ensure that every two-qubit gate operates on physically adjacent qubits.  In this context, Haar isometries function as the analytical benchmark and do not contribute to the circuit count. Uniform Clifford samples form the primary plain-reseeding ensemble, while the seed-blind template represents our construction in opposition to a provider. The fixed $[[5,1,3]]$ code \cite{laflamme1996perfect} serves as the distance-3 control, and distance-preserving reseeding maintains the code's distance during reseeding. We use HSE to represent the structured ablation because it is assembled from four separable primitives, namely Hadamard, phase, permutation, and CNOT layers, each of which performs a distinct Pauli-mixing function. Because suppressing one primitive leaves the PCG64 \cite{oneill2014pcg} draw order unchanged, the ablated variants of a given seed retain every remaining instruction, and each paired difference can therefore be attributed to the omitted component. See details in Hadamard-structured ensemble section in \cref{HSE}. Families with identical physical width utilize the same labeled fault set, and our gate-level noise study evaluates HSE (\cref{sec:method}).

\paragraph{Haar isometry}
For seed $s$, we select independent real variables
$x_{ij},y_{ij}\sim\mathcal{N}(0,1)$ and form the $N\times K$ matrix
$G_{ij}=x_{ij}+iy_{ij}$.  The reduced QR factorization writes this matrix as
$G=QR$, where $Q$ is an $N\times K$ matrix with orthonormal columns and $R$
is a $K\times K$ upper-triangular matrix, whose entries below the diagonal
vanish.  The columns of $Q$ already form an isometry, but the factorization
fixes their phases by a convention that biases the distribution.
Multiplying column $j$ of $Q$ by the phase $R_{jj}^*/|R_{jj}|$ removes this
bias and yields the Haar-distributed isometry $V_s$, a zero diagonal entry
receiving phase one~\cite{mezzadri2007generate}.  Thus $(n,k,s)$ determines every Haar
sample.  We construct the $N\times K$ isometry directly and omit an
$N\times N$ unitary circuit.

\paragraph{Uniform Clifford}
The seeded uniform Clifford sampler of Qiskit~\cite{javadiabhari2024quantum}
produces an $n$-qubit Clifford $U_s$.  Let $W$ be the map from
$\mathbb{C}^K$ to $\mathbb{C}^N$ that appends $n-k$ ancillas in $|0\rangle$.
The logical qubits have circuit qubit indices $0,\ldots,k-1$, and the
ancillas have indices $k,\ldots,n-1$.  With basis states ordered as
$|q_{n-1}\cdots q_0\rangle$, $W$ maps each logical basis state $|x\rangle$
to $|0^{n-k}\rangle|x\rangle$, whose full-register integer index is also
$x$.  Thus $W$ is the $N\times K$ matrix formed by stacking $I_K$ above
zeros, and we obtain $V_s=U_sW$ by selecting columns $0,\ldots,K-1$ of
$U_s$.  For this reference ensemble, we measure the fraction of sampled
encoders that are exactly correcting, as defined in the matrix audit below.

\paragraph{Fixed QEC control}
The standard $[[5,1,3]]$ code corrects any one-qubit
fault~\cite{laflamme1996perfect}.  Its distance is $D=3$, so that
$c_V(E)=0$ for every Pauli of weight 1 or 2 by \cref{prop:distance}.
Weight $D=3$ contains logical operators, and a seed-aware attacker at
$t\geq D$ therefore finds an accepted corruption on every seed.  Since the code does not change with the seed, that fault
transfers to every later seed as well.  For the gate count we synthesize one
encoder from its stabilizer description by Gaussian elimination in Stim
and then freeze that circuit.
The synthesis contains 21 two-qubit instructions.  The $t=1$ matrix audit
uses the exact encoding matrix.

\paragraph{Distance-preserving reseeding}
This family applies the fixed encoder, then a fresh single-qubit Clifford on
every qubit, and finally a uniform qubit permutation.  Single-qubit Cliffords
and permutations preserve Pauli weight, so the distance stays $D=3$ and
faults of weight 1 and 2 satisfy $c_V(E)=0$.  The seed changes which
weight-3 Paulis are logical.  The median two-qubit count at $n=5$ is 24,
because the permutation adds a small number of SWAP instructions to the
21-gate encoder.

\paragraph{Hadamard-structured ensemble}
\label{HSE}
An HSE circuit consists of $d$ rounds of single-qubit basis mixing followed
by random permutation and a parallel entangling matching, and decoding applies the same gates in reverse order with their inverses.  Each round draws a
Hadamard mask, the subset of qubits that receive a Hadamard gate, phase
exponents, a uniform qubit permutation realized by SWAP instructions, and a
uniform CNOT matching, a set of disjoint qubit pairs that each receive a
CNOT.  All draws come from one stream of the PCG64 pseudorandom number
generator~\cite{oneill2014pcg} initialized from $s$.  Hadamard gates exchange
$X$ and $Z$ support, phase gates mix $X$ and $Y$, permutations relocate
support, and CNOT matchings distribute it across qubits.  Every HSE encoder
is a Clifford circuit, so stabilizer simulation is exact on it, and
permutation SWAPs may act on nonadjacent qubits, so restricted connectivity
requires a separate routing step.  The ablation study of
\cref{sec:results} gives the exact draw order and the component ablations.

\paragraph{Seed-blind template}
The template of \cref{eq:template} places $L$ CZ layers on a fixed
interaction pattern.  For nearest-neighbour hardware, we use a brickwork on
a linear chain, whose odd layers pair qubits $(0,1),(2,3),\ldots$ and whose
even layers pair $(1,2),(3,4),\ldots$.  For all-to-all hardware such as
trapped ions, every layer is a perfect matching, a partition of the qubits into disjoint
pairs, drawn once from a public
PCG64 stream with the constant seed $12345+n$, which is independent of every
client seed; for odd $n$ one qubit stays unpaired in each layer.  Each
$C_j(s)$ draws one of the 24 single-qubit Cliffords for every qubit from a
PCG64 stream initialized from $s$.  The template therefore contains at most
$L\lfloor n/2\rfloor$ CZ gates, its two-qubit structure is identical for
every seed, and the seed determines the single-qubit gates and no two-qubit
gate, so $B$ controls single-qubit pulses. 

\paragraph{Matrix audit}
We compute every metric of \cref{sec:formal} directly from
the encoder, without simulating measurement-based quantum circuits.  For each encoder, we store the
isometry $V_s$ as an explicit $N\times K$ complex matrix.  For every fault
$E$ in \cref{eq:error-set}, we form the $K\times K$ matrix
$A_E=V_s^\dagger EV_s$ and evaluate the acceptance and the accepted
corruption by \cref{eq:pacc,eq:corr}.  For every pair of faults, we form
$B_{ab}$ and evaluate the residuals of \cref{eq:kl-residual}, and we also compute
the recovery fidelity $F_e$ of Appendix~\ref{app:petz}.  These steps
consist of matrix products and traces in double precision, so each
value is exact up to rounding, whereas a circuit simulation would estimate
it from a finite number of shots.  Over the fault set, we report the mean
and maximum acceptance, the mean and maximum accepted corruption, and the
maximum and root-mean-square residuals.  We label an encoder exactly
correcting when $R_{\max}\leq10^{-10}$, meaning that a recovery can correct
the entire tested fault set.  We separately label it postselection-safe when
its worst accepted disturbance satisfies $C_{\mathrm{aware}}\leq10^{-10}$,
meaning that no tested fault can both pass the acceptance check and alter
the logical state.  The first label tests all fault-pair conditions required
for correction, whereas the second tests whether each individual fault is
rejected or harmless when accepted.

\paragraph{Gate-level audit}
For each encoder and fault, we run one circuit with the logical state
$|0\rangle$ and another with $|+\rangle$.  We encode the state with $U_s$,
inject the fault, decode with $U_s^\dagger$, and measure every qubit.
Before measuring the $|+\rangle$ circuit, we rotate its logical qubit into
the computational basis.  We accept a shot only when every ancilla returns
zero, and we mark an accepted shot as corrupted when the logical result
differs from the prepared state.  We record the accepted and corrupted-shot
totals together with every raw output string.  For each input, we define the
corruption rate as the number of corrupted accepted shots divided by the
total number of shots.  In an ideal Clifford circuit with a Pauli injection,
the larger of the two rates equals $c_V(E)$, because any logical $X$, $Y$,
or $Z$ error changes at least one input state with certainty.

\section{Experimental Methodology}
\label{sec:method}

We aim to determine whether the implementation replicates Haar randomness and to assess the proximity of each executable family. We conduct the matrix audit described in \cref{sec:encoders} on Haar, uniform Clifford, and HSE encoders. Given that the matrix audit is precise, the mean over seeds fluctuates with the sampled encoders and can be compared with \cref{eq:haar-acc,eq:haar-corr}. The $[[5,1,3]]$ code, which rectifies every single-qubit fault, serves to verify the audit itself.

\paragraph{Attacker-knowledge study}
How much does hiding the seed reduce corruption?  We split the seeds into
disjoint training and test sets.  The seed-blind attacker selects its fault
on the training seeds (\cref{eq:blind-select}), and we score that fault and
the seed-aware attacker on the test seeds
(\cref{eq:blind-test,eq:aware-test}).  The difference between the two scores
is the benefit of hiding the seed, and a uniformly random fault serves as a
non-adaptive baseline.

\paragraph{Weight-budget study}
Does reseeding help once the fault budget reaches the code distance?  We
repeat the attacker-knowledge protocol at $t=3$ and $t=4$, which reach and
exceed $D=3$, for the fixed code, distance-preserving reseeding, uniform
Clifford, and HSE, and we add the placement-adaptive attacker.  At these
budgets, the quantum Hamming bound, which limits how many faults a code of
a given size can correct, rules out exact correction at our widths.  We
therefore report postselection safety and fixed-fault transfer.

\paragraph{Seed-blind study}
At what depth does the template reach the uniform Clifford rate?  We
evaluate the support chain of \cref{eq:chain} for both interaction patterns
across widths and template depths. We check the chain against Stim simulations of sampled templates
and then run the weight-budget protocol on the template, where the
seed-aware score represents a provider that would learn $s$ without $B$.

\paragraph{Gate-level study}
Does integrity survive noisy gates?  A noisy circuit is the encode, fault,
and decode circuit in which every encoder and decoder gate is followed by
depolarizing noise.  After a one-qubit gate, a uniformly random Pauli $X$,
$Y$, or $Z$ acts on its qubit with probability $p_1=10^{-3}$, and after a
two-qubit gate, a uniformly random nonidentity two-qubit Pauli acts on the
pair with probability $p_2=10^{-2}$.  We run the gate-level audit of
\cref{sec:encoders} on HSE circuits in Stim, first without noise to
cross-check the matrix audit and then as noisy circuits, each SWAP being
decomposed into CNOTs.  The injected fault and the basis changes remain
noiseless, so any change is due to the encoder circuit. 

\paragraph{Ablation study}
Which HSE component matters?  At the largest exact width and the deepest
HSE depth, we remove one component at a time, namely Hadamard, phase,
permutation, or CNOT gates, and reuse the same seeds and the remaining
random draws, so that each difference reflects the removed component.

In our experiments, we present 95\% percentile-bootstrap intervals, which are derived by resampling seeds with replacement. For attacker scores, we resample the test seeds while maintaining the selected fault constant. The reported HSE depth is determined by the highest exact-correction fraction, with ties resolved first by the worst-fault corruption and subsequently by the two-qubit cost.

\section{Results}
\label{sec:results}

First, we confirm that our framework accurately replicates the Haar prediction and determine the precise cost associated with reseeding. Next, we demonstrate that reseeding effectively prevents a replayed fault when the attacker lacks knowledge of the seed. Finally, we establish that seed-blind reseeding extends this protection to the provider at a minimal cost.

\begin{figure}[t]
  \centering
  \includegraphics[width=\columnwidth]{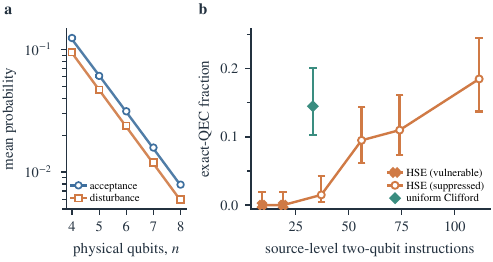}
  \caption{Matrix-audit results for $k=1$ over \ResultSeeds{} seeds per condition.
  (a) Acceptance $p_{\mathrm{acc},V}(E)$ and accepted corruption $c_V(E)$
  of \cref{eq:pacc,eq:corr}, averaged over every single-qubit Pauli fault
  and over \ResultSeeds{} Haar encoders (markers), together with the closed
  forms of \cref{thm:haar} (lines).  Both decrease as $2^{-n}$, so that each
  added qubit halves the probability that a fault chosen without the seed
  corrupts an accepted execution.  (b) Fraction of $n=\ResultNMax$
  encoders that correct every tested single-qubit fault, plotted against the
  median source-level two-qubit instruction count.  Connected points are HSE
  depths and the diamond is uniform Clifford.  Crosses mark a median
  worst-fault disturbance above $0.5$, and open circles mark values below
  it.  Bars show 95\% bootstrap intervals in (a) and Wilson intervals in
  (b).}
  \label{fig:dense-results}
\end{figure}
\begin{table}[t]
  \centering
  \caption{Executable encoders at $n=8$ over 200 seeds.
  Exact is the percentage that correct every tested single-qubit fault, and
  Safe is the percentage that are postselection-safe on that same fault set.
Petz fidelity and pre-transpilation two-qubit instruction
  counts are medians.}
  \label{tab:dense-results}
  \begin{tabular}{@{}lrrrr@{}}
\toprule
Encoder & Petz $F_e$ & Exact (\%) & Safe (\%) & 2Q \\
\midrule
Clifford & 0.920 & 14.5 & \namark & 33 \\
\ensuremath{\mathrm{HSE}_{1}} & 0.784 & 0.0 & 0.0 & 9 \\
\ensuremath{\mathrm{HSE}_{2}} & 0.800 & 0.0 & 12.0 & 19 \\
\ensuremath{\mathrm{HSE}_{4}} & 0.867 & 1.5 & 52.5 & 37 \\
\ensuremath{\mathrm{HSE}_{6}} & 0.920 & 9.5 & 66.5 & 56 \\
\ensuremath{\mathrm{HSE}_{8}} & 0.920 & 11.0 & 78.0 & 74 \\
\ensuremath{\mathrm{HSE}_{12}} & 0.920 & 18.5 & 84.5 & 112 \\
\bottomrule
\end{tabular}
\end{table}

\paragraph{Haar reference}
The measured acceptance and accepted corruption follow the closed-form Haar
curves (\cref{fig:dense-results}(a)), which decay as $2^{-n}$ because
$N=2^n$ doubles with each added qubit while $K=2$ stays fixed.  Corruption
remains below acceptance because part of an accepted fault acts as a phase
and leaves the logical state unchanged.

\paragraph{Exact correction and postselection}
We note that reseeding is associated with a cost, as random encoders rarely achieve precise fault correction (\cref{tab:dense-results}). Even when the width is matched at $n=5$, no HSE encoder is capable of correcting every single-qubit fault. In contrast, the fixed $[[5,1,3]]$ code is able to correct all such faults, as assured by \cref{prop:distance}. Postselection safety is less reliable than exact correction, as it only requires that each fault be either rejected or rendered harmless, rather than ensuring that every pair of faults satisfies \cref{eq:kl}. For every $t=1$ metric, HSE aligns with uniform Clifford but incurs different costs (\cref{fig:dense-results}(b)).

\begin{figure*}[t]
  \centering
  \includegraphics[width=\textwidth]{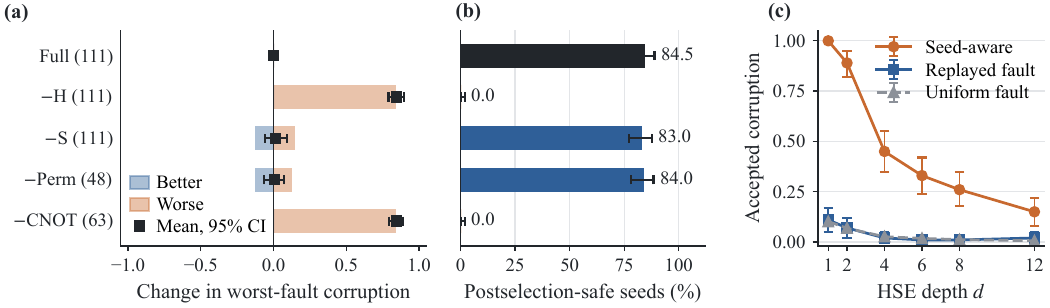}
  \caption{HSE ablation and attacker knowledge at $t=1$ and $n=8$. Panels (a) and (b) compare each ablated variant with the complete HSE at depth 12 over the same 200 seeds, and the parenthesized number after each variant is its mean count of two-qubit instructions. (a) Change in worst-fault accepted corruption relative to the complete HSE. Because this change is $-1$, $0$, or $+1$ for each seed, the bars give the fractions of seeds that become better or worse, and the markers give the mean change. (b) Percentage of seeds that remain postselection-safe. (c) Held-out accepted corruption of the seed-aware attacker, of the replayed fault selected on the training seeds, and of a uniformly random fault, which serves as the reference for the replayed fault.}
  \label{fig:hse-ablation}
  \label{fig:attack-knowledge}
\end{figure*}

\paragraph{Ablation study}
In the context of HSE, the Hadamard and CNOT gates are integral to maintaining protection; the removal of either gate results in the loss of postselection safety, as illustrated in \cref{fig:hse-ablation}(a,b). 
Formally, HSE$(n,d;s)$ contains $d$ independently sampled layers, supplied
by one PCG64 pseudorandom stream~\cite{oneill2014pcg} initialized
from $s$ in the following fixed draw order.
\begin{enumerate}
  \item $n$ Bernoulli$(1/2)$ bits form a Hadamard mask, and an empty mask
        sets the bit of one uniformly drawn qubit.
  \item Each qubit in ascending index order receives an exponent
        $r\in\{0,1,2,3\}$ drawn uniformly and emitted as $r$ phase
        instructions.
  \item A uniform permutation of the $n$ qubits is realized by successive
        transpositions, which emit at most $n-1$ SWAP instructions.
  \item A uniform ordering of the $n$ qubits pairs consecutive entries into
        $\lfloor n/2\rfloor$ CNOT instructions, with one uniform bit per
        pair selecting the control.
\end{enumerate}
For odd values of $n$, the matching process leaves the final element of that
ordering unmatched.  The triple $(n,d,s)$ identifies an HSE sample, and
depths that share a seed are nested: the first $d$ layers of
HSE$(n,12;s)$ coincide with HSE$(n,d;s)$, because both consume one PCG64
stream from the first layer.  Suppressing the emission of one primitive
leaves the draw order unchanged, so the ablation variants of one seed share
every remaining instruction.  Eliminating Hadamard or CNOT gates functions as a sanity check, as these components can leave a qubit susceptible to undetected logical faults, unlike the removal of permutations. 

We show that at $(n,k,d)=(\AblationN,\AblationK,\AblationDepth)$, \AblationFullSafe\% of
full-HSE seeds are postselection-safe in \cref{tab:dense-results}.  Removing either Hadamard or CNOT
gates reduces this fraction to zero in the sampled seeds. Here, Hadamard gates exchange $X$ and
$Z$ support, while CNOT gates can propagate Pauli support between logical
and ancilla qubits.  Without CNOT gates, the remaining single-qubit gates and
permutations do not distribute a logical qubit across several physical
qubits, so that a Pauli fault on the qubit carrying that qubit can change the
logical state without changing an ancilla.  Without Hadamard gates, the
encoder maps computational basis states to computational basis states up to
phase, so that a physical $Z$ on a qubit whose value depends on the logical
bit induces an undetected logical phase error. 
In contrast, removing phase gates and permutations changes which seeds are
postselection-safe, but we note that the paired comparison resolves no aggregate loss at
this condition.  The H and CNOT are sanity checks, as we stated
in \cref{sec:encoders}.

\paragraph{Attacker knowledge at $t=1$} 
Reseeding effectively mitigates replayed faults by ensuring precise correction. Concealing the seed substantially diminishes the extent of accepted corruption (\cref{fig:attack-knowledge}(c)), and the fixed-fault score associated with uniform Clifford reseeding corresponds with the rate specified in \cref{cor:clifford}. This reduction is attributed to rejection, as acceptance results in corruption in these scenarios. Nevertheless, a provider with access to the seed retains the seed-aware score.

\begin{table}[t]
  \centering
  \caption{Number of the 100 test seeds on which a seed-blind fault
  corrupts an accepted execution.  Replayed applies the single fault
  selected on the training seeds to every test seed, and Adaptive selects a
  fresh set of $t$ qubits on each test seed (\cref{sec:threat}).  A
  seed-aware attacker corrupts all 100 test seeds in every row.  2Q is the
  median number of two-qubit instructions under the rule in
  \cref{sec:encoders}.}
  \label{tab:weight-budget}
  \footnotesize
  \setlength{\tabcolsep}{2.2pt}
  \begin{tabular}{@{}lrrrrrr@{}}
\toprule
 & & & \multicolumn{2}{c}{$t{=}3$} & \multicolumn{2}{c}{$t{=}4$} \\
\cmidrule(lr){4-5}\cmidrule(l){6-7}
Encoder & $n$ & 2Q & Replayed & Adaptive & Replayed & Adaptive \\
\midrule
Fixed $[[5,1,3]]$ & 5 & 21 & 100 & 100 & 100 & 100 \\
\shortstack[l]{Distance-\\preserving} & 5 & 24 & 19 & 13 & 19 & 12 \\
HSE, $d{=}12$ & 5 & 57 & 7 & 5 & 7 & 5 \\
Uniform Clifford & 5 & 13 & 2 & 6 & 5 & 3 \\
HSE, $d{=}12$ & 8 & 112 & 1 & 1 & 0 & 0 \\
Uniform Clifford & 8 & 33 & 2 & 1 & 1 & 0 \\
\bottomrule
\end{tabular}
\end{table}

\begin{figure}[th]
  \centering
  \includegraphics[width=\columnwidth]{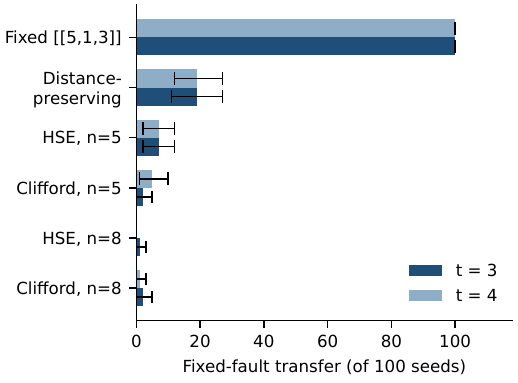}
  \caption{Fixed-fault transfer at $t=3$ and $t=4$.
  Bars are counts out of 100 test seeds, and whiskers are 95\% bootstrap
  intervals on those seeds with the training fault held fixed.  The
  seed-aware count is 100 for every bar.}
  \label{fig:weight-budget}
\end{figure}

\paragraph{Weight budget $t=3$ and $t=4$} 
The observed pattern remains consistent when the budget reaches the code distance (\cref{tab:weight-budget,fig:weight-budget}). At this juncture, each encoder contains a logical operator within the budget, allowing the seed-aware attacker to compromise every test seed of each encoder. In the absence of the seed, the encoders diverge: the fixed code consistently replicates its learned fault on every seed, whereas uniform Clifford reseeding transfers it at a rate consistent with \cref{eq:clifford-bound}. Distance-preserving reseeding transfers more frequently than HSE while maintaining distance 3, rendering it preferable when precise correction below the distance is required.

\paragraph{Seed-blind reseeding} 
In hardware configurations featuring all-to-all connections, a shallow template achieves the uniform Clifford rate. Conversely, in a linear chain, the necessary depth increases approximately linearly with $n$, as each CZ layer shifts a fault by at most one qubit (\cref{tab:seed-blind}). We observe that a larger budget provides minimal benefit to a seed-blind attacker. The results for the chain are consistent with Stim across all \SBCheckCases{} validation cases, within \SBCheckMaxSE{} binomial standard errors.

\begin{table}[t]
  \centering
  \caption{Seed-blind template.  For each width $n$, the shallowest depth
  $L$ whose worst seed-blind fault of weight at most 3 lies within
  \SBFactorPct\% of the uniform Clifford rate of \cref{cor:clifford}, with its CZ count.
  Each rate is one corrupted execution in the stated number.}
  \label{tab:seed-blind}
  \small
  \setlength{\tabcolsep}{4pt}
  \begin{tabular}{@{}rr rrr rrr@{}}
\toprule
 & & \multicolumn{3}{c}{All-to-all} & \multicolumn{3}{c}{Linear chain} \\
\cmidrule(lr){3-5}\cmidrule(l){6-8}
$n$ & Clifford & $L$ & CZ & Worst & $L$ & CZ & Worst \\
\midrule
5 & 21 & 10 & 20 & 20 & 19 & 38 & 20 \\
6 & 43 & 11 & 33 & 41 & 23 & 58 & 39 \\
7 & 85 & 14 & 42 & 82 & 29 & 87 & 79 \\
8 & 171 & 11 & 44 & 165 & 35 & 123 & 160 \\
10 & 683 & 12 & 60 & 629 & 45 & 203 & 638 \\
12 & 2,731 & 13 & 78 & 2,625 & 53 & 292 & 2,492 \\
14 & 10,923 & 15 & 105 & 10,189 & $>60$ & -- & -- \\
\bottomrule
\end{tabular}
\end{table}

\paragraph{Gate-level acceptance} 
The encoders are executed as noisy circuits. In ideal HSE circuits, a propagated Pauli is rejected if it possesses an $X$ or $Y$ factor on a decoded ancilla (\cref{fig:stim-results}(a)). The matrix audit and ideal Stim concur on all \StimIdealAgreementCases{} overlapping cases. However, noise rapidly reduces acceptance, resulting in the outputs of deep or wide circuits lacking logical information (\cref{fig:stim-results}(b) and \cref{tab:stim-acceptance}). This limitation constrains the circuit size at which postselection remains effective.

\begin{figure}[t]
  \centering
  \includegraphics[width=\columnwidth]{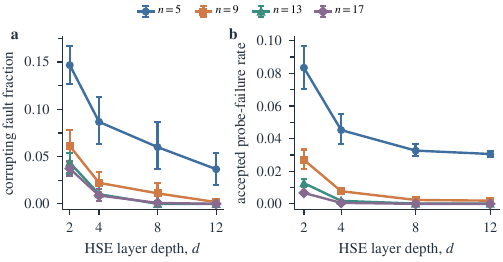}
  \caption{Gate-level HSE results over \StimSeeds{} seeds per condition.
  (a) Fraction of single-qubit Pauli faults that cause accepted corruption in
  ideal circuits.
  (b) Fraction of noisy shots that are both accepted and corrupt under
  gate-local depolarization with $(p_1,p_2)=(10^{-3},10^{-2})$.
  Panel (b) is the shot frequency under that mixture and not $c_V(E)$ for
  one chosen Pauli.  Points are seed means and bars are 95\% bootstrap
  intervals.}
  \label{fig:stim-results}
\end{figure}

\begin{table}[t]
  \centering
  \caption{Noisy HSE acceptance.
  Joint is the fraction of shots that are accepted and corrupt, and
  Conditional is $P(\mathrm{corrupt}\mid\mathrm{accept})$.
  Rates use $p_1=10^{-3}$, $p_2=10^{-2}$, 20 seeds, and 1000 shots per
  circuit, and each seed pools every weight-1 Pauli and both probe states.}
  \label{tab:stim-acceptance}
  \scriptsize
  \begin{tabular}{@{}rrrrr@{}}
\toprule
$n$ & $d$ & Acceptance & Joint & Conditional \\
\midrule
5 & 2 & $2.31\times 10^{-1}$ & $8.35\times 10^{-2}$ & 0.36 \\
5 & 4 & $1.11\times 10^{-1}$ & $4.52\times 10^{-2}$ & 0.41 \\
5 & 8 & $6.84\times 10^{-2}$ & $3.27\times 10^{-2}$ & 0.48 \\
5 & 12 & $6.34\times 10^{-2}$ & $3.06\times 10^{-2}$ & 0.48 \\
9 & 2 & $9.88\times 10^{-2}$ & $2.70\times 10^{-2}$ & 0.27 \\
9 & 4 & $2.23\times 10^{-2}$ & $7.76\times 10^{-3}$ & 0.35 \\
9 & 8 & $5.04\times 10^{-3}$ & $2.39\times 10^{-3}$ & 0.47 \\
9 & 12 & $4.00\times 10^{-3}$ & $1.97\times 10^{-3}$ & 0.49 \\
13 & 2 & $5.89\times 10^{-2}$ & $1.27\times 10^{-2}$ & 0.22 \\
13 & 4 & $7.54\times 10^{-3}$ & $1.86\times 10^{-3}$ & 0.25 \\
13 & 8 & $4.07\times 10^{-4}$ & $1.65\times 10^{-4}$ & 0.40 \\
13 & 12 & $2.83\times 10^{-4}$ & $1.35\times 10^{-4}$ & 0.48 \\
17 & 2 & $3.57\times 10^{-2}$ & $6.71\times 10^{-3}$ & 0.19 \\
17 & 4 & $2.96\times 10^{-3}$ & $6.08\times 10^{-4}$ & 0.21 \\
17 & 8 & $3.97\times 10^{-5}$ & $1.42\times 10^{-5}$ & 0.36 \\
17 & 12 & $1.23\times 10^{-5}$ & $5.88\times 10^{-6}$ & 0.48 \\
\bottomrule
\end{tabular}
\end{table}

\section{Conclusion}
\label{sec:conclusion}

This study introduces an auditing per-run encoder reseeding framework for use on a shared cloud processor. A fixed quantum error correction (QEC) code permits a single learned fault to compromise each execution, while plain reseeding reduces the replay rate by half with each additional qubit, provided attackers do not observe the seed. Seed-blind reseeding further extends this protection to the provider by employing a public controlled-Z (CZ) template with seed-selected single-qubit gates, implemented by a trusted control-plane boundary, thereby ensuring the provider adheres to the co-tenant's rate. With \SBFourteenDepth{} layers at $n=14$, each seed-blind fault affects 1 in \SBFourteenOneIn{} executions. Future work will focus on hardware validation.

\bibliographystyle{IEEEtran}
\bibliography{refs}

\appendices

\section{Proof of the Haar Expectation Theorem}
\label{app:proof}

Let $P_0=\sum_{j=1}^{K}|j\rangle\!\langle j|$ be a fixed rank-$K$ projector
on $\mathbb{C}^{N}$.  A uniformly random rank-$K$ projector can be written
as $P(U)=UP_0U^\dagger$, where $U$ is Haar distributed on $U(N)$.  Hence its
second moment is
\[
 M\equiv\E[P\otimes P]
   =\int_{U(N)}P(U)\otimes P(U)\,dU.
\]
For any fixed $W\in U(N)$, the left invariance of the Haar measure and the
change of variables $U'=WU$ give
\[
 \begin{aligned}
 &(W\otimes W)M(W^\dagger\otimes W^\dagger)\\
 &\quad=\int_{U(N)}P(U')\otimes P(U')\,dU'=M.
 \end{aligned}
\]
Thus $M$ commutes with $W\otimes W$ for every $W\in U(N)$.  The tensor space
decomposes as
\[
 \mathbb{C}^{N}\otimes\mathbb{C}^{N}
 =\operatorname{Sym}^2(\mathbb{C}^{N})
  \oplus\bigwedge^2(\mathbb{C}^{N}),
\]
the two irreducible subspaces of the $W\otimes W$ action, with projectors
$\Pi_+=(I_{N^2}+F)/2$ and $\Pi_-=(I_{N^2}-F)/2$.  Here $F$ is the swap
operator defined by $F(|x\rangle\otimes|y\rangle)=|y\rangle\otimes|x\rangle$.
Schur's lemma therefore gives $M=\lambda_+\Pi_++\lambda_-\Pi_-$.  Setting
$a=(\lambda_++\lambda_-)/2$ and $b=(\lambda_+-\lambda_-)/2$ yields
\begin{equation}
 \E[P\otimes P]=aI_{N^2}+bF,
 \label{eq:projector-ansatz}
\end{equation}
where $a$ and $b$ are real coefficients.  Because $P$ has rank $K$ and is a
projector, $\Tr(P)=\Tr(P^2)=K$.  Taking the trace of
\cref{eq:projector-ansatz}, first directly and then after multiplication by
$F$, gives the two linear equations
\begin{equation}
 aN^2+bN=K^2,
 \qquad
 aN+bN^2=K.
 \label{eq:projector-constraints}
\end{equation}
Solving \cref{eq:projector-constraints} yields
\begin{equation}
 a=\frac{K(KN-1)}{N(N^2-1)},\qquad
 b=\frac{K(N-K)}{N(N^2-1)}.
 \label{eq:projector-moment}
\end{equation}
For any $N\times N$ matrices $X$ and $Y$, the swap identity is
$\Tr[(X\otimes Y)F]=\Tr(XY)$.  Applying it to the fixed traceless unitary
$E$ gives the two contractions
\begin{equation}
\begin{aligned}
 \Tr(PE^\dagger PE)
   &=\Tr[(P\otimes P)(E^\dagger\otimes E)F],\\
 |\Tr(PE)|^2
   &=\Tr[(P\otimes P)(E\otimes E^\dagger)].
\end{aligned}
 \label{eq:haar-contractions}
\end{equation}
Substituting \cref{eq:projector-ansatz} into the first contraction, the
$aI_{N^2}$ term contributes $a\Tr(E^\dagger E)=aN$, while the $bF$ term
contributes $b|\Tr(E)|^2=0$.  Since
$p_{\mathrm{acc},V}(E)=K^{-1}\Tr(PE^\dagger PE)$ by \cref{eq:pacc}, we
obtain
\begin{equation}
\begin{aligned}
 \E\,p_{\mathrm{acc},V}(E)
 &=K^{-1}(aN)\\
 &=\frac{KN-1}{N^2-1}.
\end{aligned}
 \label{eq:proof-haar-acceptance}
\end{equation}
For the second contraction, the $aI_{N^2}$ term again vanishes because it
equals $a|\Tr(E)|^2$, and the $bF$ term contributes $b\Tr(EE^\dagger)=bN$.
Therefore
\begin{equation}
\begin{aligned}
 \E|\Tr(A_E)|^2
 &=\E|\Tr(PE)|^2\\
 &=bN
 =\frac{K(N-K)}{N^2-1}.
\end{aligned}
 \label{eq:proof-haar-trace}
\end{equation}
Finally, taking expectations in \cref{eq:decomposition}, substituting
\cref{eq:proof-haar-acceptance,eq:proof-haar-trace}, and simplifying gives
\begin{equation}
\begin{aligned}
 \E c_V(E)
 &=\frac{KN-1}{N^2-1}
   -\frac{N-K}{K(N^2-1)}\\
 &=\frac{N(K^2-1)}{K(N^2-1)}.
\end{aligned}
 \label{eq:proof-haar-disturbance}
\end{equation}
These expressions coincide with \cref{eq:haar-acc,eq:haar-corr}.

\section{Fixed Petz Recovery}
\label{app:petz}

We measure approximate-QEC performance with a reproducible transpose
(Petz) recovery~\cite{barnum2002reversing,beny2010general}.  To this end, we
reuse the supported list $\{E_a\}_{a=0}^{q}$, set $m=q+1$, and assign each
operator the same probability $1/m$.  The physical noise channel and its
Kraus operators $F_a$ are therefore
\begin{equation}
 F_a=\frac{E_a}{\sqrt m},
 \qquad
 \cN(\omega)=\sum_{a=0}^{q}F_a\omega F_a^\dagger,
 \label{eq:uniform-channel}
\end{equation}
where $\omega$ is any physical $N\times N$ density matrix.  Evaluating this
channel on the maximally mixed code state $\rho=P/K$ gives
$\sigma=\cN(\rho)$.  If $\sigma=\sum_j\lambda_j|j\rangle\!\langle j|$, its
support-restricted inverse is
$\sigma^{-1/2}=\sum_{j\mid\lambda_j>0}\lambda_j^{-1/2}|j\rangle\!\langle j|$.
In double-precision evaluation, we retain an eigenvalue if and only if it
exceeds $10^{-12}\max\{\lambda_{\max},1\}$, and we fix this numerical
support threshold in the implementation.  Composing the transpose recovery
with $V^\dagger$ gives the logical recovery Kraus operators
\begin{equation}
 L_a=K^{-1/2}V^\dagger F_a^\dagger\sigma^{-1/2},
 \qquad 0\leq a\leq q.
 \label{eq:petz-kraus}
\end{equation}
For noise branch $b$ followed by recovery branch $a$, the effective logical
Kraus operator is $M_{ab}=L_aF_bV$.  The entanglement fidelity of the
recovered logical channel is then
\begin{equation}
 F_e=K^{-2}\sum_{a,b=0}^{q}|\Tr(M_{ab})|^2
    =K^{-2}\sum_{a,b=0}^{q}|\Tr(L_aF_bV)|^2.
 \label{eq:petz}
\end{equation}
Here $0\leq F_e\leq1$, and $F_e=1$ means that the recovered channel
preserves half of a maximally entangled logical state without error.  We
recompute the recovery from the $V$ and $\sigma$ of each sample.

\section{Open Science}
\label{app:open-science}

Our implementation contains seeded ensemble generators, exact and Stim
runners, analysis scripts, checked result files, and deterministic unit
tests.  The canonical entry point is \texttt{artifact/run. py}, and the
top-level reproduction script provides bounded smoke and paper modes.

Every canonical run emits a manifest containing the resolved command and
seed arguments, the run time in Coordinated Universal Time, Python and
dependency versions, the Git revision when available, repository-relative
output paths, and checksums.  We retain raw integer counts and tidy derived
tables.  Figure scripts read the arrays without pickle-dependent metadata
and run headlessly.

The implementation is available through the URL given in
\cref{sec:intro}.  The supplied release builder uses an explicit file
allowlist, scans for local paths and credentials, and writes checksums for
the copied files.

\section{Ethics Considerations}

We identified no ethical concerns in this work and reached that conclusion
as follows.  The experiments operate on generated circuits and insert
standard Pauli gates in local simulation, so that no human subjects, personal
data, production systems, or third-party services are involved.  The
implementation contains circuit generation, Pauli insertion, matrix analysis,
and Stim simulation.  It contains no co-tenant interaction, external-system
modification, or pulse-level control code, and it does not provide a
working attack against any deployed quantum cloud service.  No vulnerability
disclosure is therefore required.

\section{Reproduction Checklist}

All commands run from the repository root in the pinned container, and
reproduction proceeds as follows.
\begin{enumerate}
  \item Build the execution image with
  \texttt{docker compose build artifact}.
  \item Run the bounded smoke workflow with
  \texttt{docker compose run {-}{-}rm artifact} before the full experiment.
  This command executes the deterministic unit tests, a two-seed matrix-audit
  audit, a held-out attacker-selection check, an HSE ablation, and a
  one-seed Stim integration check.  Any failed invariant terminates the
  workflow.
  \item Regenerate the declared paper grid.  The command validates and
  resumes durable per-seed checkpoints, then regenerates summaries and
  figures with \texttt{docker compose run {-}{-}rm paper}.  Before resuming
  a run, the command checks that the checkpoint is an exact ordered prefix
  of the declared parameter grid.
  \item For an independent computation from raw seeded jobs, bypass the
  checked-in checkpoints by choosing a nonexistent or empty directory and
  running the following command.
\begin{verbatim}
docker compose run --rm artifact \
  ./artifact/reproduce.sh --paper-fresh \
  artifact/results/fresh-paper
\end{verbatim}
  The complete grid uses 200 matrix-audit seeds, 200 paired ablation seeds, and the
  declared Stim sweeps.  It requires several CPU-hours, at least 8~GiB of
  RAM, and approximately 1~GiB of free disk.
  \item Recompute only the derived statistics, figures, and compact result
  bundle from completed canonical checkpoints with
  \texttt{docker compose run {-}{-}rm postprocess}.  Compare the SHA-256
  hashes in \texttt{artifact/manifests/} and in the embedded manifest of
  \texttt{artifact/results/paper-results. npz}.
\end{enumerate}

\paragraph{Equation-to-code traceability}
The following mapping permits the sampling algorithm and each mathematical
result to be reproduced independently of the plotting layer.
\begin{itemize}
  \item The HSE layer algorithm of \cref{sec:encoders} is implemented by
  \texttt{hadamard\_structured\_circuit} in
  \texttt{artifact/src/ensembles.py}.  Its keyword flags suppress the
  emission of one primitive without altering the draw order, which supplies
  the paired ablation variants.
  \item \Cref{eq:error-set} is implemented by \texttt{pauli\_errors} in
  \texttt{artifact/src/metrics.py}, and the stored Pauli strings expose the
  exact enumeration order.
  \item \Cref{eq:kl,eq:kl-residual,eq:uniform-channel,eq:petz-kraus,eq:petz}
  are evaluated by \texttt{evaluate\_isometry} in the same file.  The fixed
  five-qubit baseline checks $R_{\max}<10^{-10}$ and $1-F_e<10^{-10}$.
  \item \Cref{eq:compressed-action,eq:pacc,eq:corr,eq:decomposition} are
  evaluated per labeled fault by \texttt{fault\_integrity\_profile}.  Unit
  tests verify $0\leq c_V(E)\leq p_{\mathrm{acc},V}(E)\leq1$.
  \item \Cref{eq:haar-acc,eq:haar-corr} are implemented by
  \texttt{haar\_expected\_acceptance} and
  \texttt{haar\_expected\_accepted\_corruption}, and the matrix-audit runner
  compares these closed forms with seeded QR samples.
  \item \Cref{eq:aware,eq:blind-select,eq:blind-test,eq:aware-test} define
  seed-aware selection and fixed-fault transfer.  The $t=1$ study is
  implemented in \texttt{artifact/run\_attack\_knowledge.py}, and the $t=3$
  and $t=4$ study is implemented in
  \texttt{artifact/run\_weight\_budget.py}.  Its CSV retains the seed split,
  the selected fault, and both scores.
\end{itemize}

Each raw row records the defining parameters $\mathsf{F}$, $n$, $k$, $d$,
$s$, and $t$, and Stim rows additionally record the shot count and the noise
probabilities $p_1$ and $p_2$.  Run manifests record the resolved command,
the operating system, Python and package versions, the Git revision when
available, output paths, and SHA-256 hashes.  Consequently every reported
value can be traced from a derived table to raw rows, a parameterized
generator, and the equation that defines its metric.

\end{document}